\documentclass[]{elsarticle}
\usepackage[sc,noBBpl]{mathpazo}
\usepackage[mathscr]{eucal}
\usepackage{amsmath,amsfonts,amssymb,amsthm}
\usepackage{graphicx}
\usepackage{mathtools}
\usepackage{color}
\usepackage{tgpagella}
\usepackage{subfigure}
\usepackage{caption}
\usepackage{url}

\newcommand\mvector{\boldsymbol}
\newcommand\field{\mathbb}
\newcommand\scN{{\mathscr N}}
\newcommand\C{\field{C}}
\newcommand\vvarphi{\mvector{\varphi}}
\newcommand\Dt{\frac{\mathrm{d}\phantom{t} }{\mathrm{d}\mspace{1mu}
t}}
\newcommand\Dz{\frac{\mathrm{d}\phantom{z} }{ \mathrm{d}z}}
\newcommand\Dtt{\frac{\mathrm{d}^2\phantom{t} }{\mathrm{d}t^2}}
\newcommand\Dzz{\frac{\mathrm{d}^2\phantom{z} }{\mathrm{d}z^2}}
\newcommand\rmd{\mathrm{d}}
\newcommand\scG{{\mathscr G}}
\newcommand\scT{{\mathscr T}}
\newcommand\scD{{\mathscr D}}
\newcommand\Q{\field{Q}}
\newcommand\Z{\field{Z}}
\newcommand\N{\field{N}}
\newcommand\vl{\mvector{l}}
\newcommand\rmi{\mathrm{i}\mspace{1mu}}
\newcommand\vM{\mvector{M}}

\theoremstyle{plain}
\newtheorem{theorem}{Theorem}
\newtheorem{lemma}{Lemma}
\newtheorem{remark}{Remark}

\newtheoremstyle{note}{\topsep}{\topsep}{\slshape}{}{\scshape}{}{ }{}
\theoremstyle{note}

\numberwithin{equation}{section}
\numberwithin{theorem}{section}
\numberwithin{lemma}{section}
\numberwithin{proposition}{section}
\numberwithin{corollary}{section}
\numberwithin{remark}{section}

\mathtoolsset{mathic,centercolon}

\begin{document}

\begin{frontmatter}

\title{Non-integrability of restricted double pendula}

\author{Tomasz Stachowiak}
\ead{stachowiak@cft.edu.pl}
\address{Center for Theoretical Physics PAS, Al. Lotnikow 32/46, 02-668 Warsaw, Poland}

\author{Wojciech Szumi\'nski}
\ead{uz88szuminski@gmail.com}
\address{Institute of Physics, University of Zielona G\'ora, Licealna 9, PL-65-407, Zielona G\'ora, Poland}

\begin{abstract}
    We consider two special types of double pendula, with the motion of masses
restricted to various surfaces. In order to get quick insight into the
dynamics of the considered systems the Poincar\'e cross sections as well as bifurcation
diagrams have been  used. The numerical computations show that both
models are chaotic which suggest that they are not integrable. We give an
analytic proof of this fact checking the properties of the differential Galois
group of the system's variational equations along a particular non-equilibrium
solution.
\end{abstract}


\end{frontmatter}

\section{Introduction}
The complicated  behaviour of the simple double pendulum is well know but still
fascinating--because it is probably the simplest mechanical system exhibiting
chaos. We can find many interesting information about its behaviour on
the web~\cite{Nathan:10::}, in books~\cite{Baker:05::,Gitterman:10::} as well
as in the scientific
articles~\cite{Stachowiak:06::,Ivanov:99::,Burov:02::,Dullin:94::}. Since
such a system depends on parameters, like
masses and lengths of arms, its integrability analysis is very difficult. Until
now, there is no closed mathematical proof confirming its non-integrability.
However, in the nineties of the XX century, Morales-Ruiz and Ramis showed that 
integrability in the Liouville sense imposes a very restrictive condition for
the identity component of the differential Galois group of variational equations
obtained by the lineralization of the equations of motion along a particular
non-equilibrium solution. The main theorem of this theory states that if the
system is integrable in the Liouville sense, then the identity component of
differential Galois group of normal variational equations is Abelian. Note that
this is a necessary condition but not a sufficient one. For the
precise definition of Morales--Ramis theory and differential Galois group see
e.g. \cite{Morales:99::,Put:03::,Weil:01::}.

The weak point of the Morales--Ramis theory  is that it requires the knowledge
of a particular solution that is not an equilibrium position. Unfortunately,
for the ordinary planar double pendulum we cannot find it, while for the fully
three-dimensional double pendulum, the solutions are too simple to provide any
restrictions. Thus, in order to solve this
inconvenience we will consider some special types of double pendula restricted to
certain surfaces, for which the particular non-equilibrium solutions are known.
Thanks to this delicate modification we are able to proceed with the whole
integrability analysis of these models, so that we can move closer to the proof of
the non-integrability of the double pendulum than ever before.

\section{Model 1: simple broken pendulum \label{sec:model_1}}
The first model under consideration consists of two simple pendula such that the
first one with length $l_1$ and mass $m_1$ is attached to a pivot point and moves in
the $(x,y)$ plane. The second pendulum with $l_2$ and $m_2$ is restricted to
the plane containing $m_1$ and parallel to the $(x,z)$ plane. Looking at the
Figure~\ref{fig:m_1:geometry} we can immediately
write the Lagrange function
\begin{figure}[h]
\centering{
\resizebox{77mm}{!}{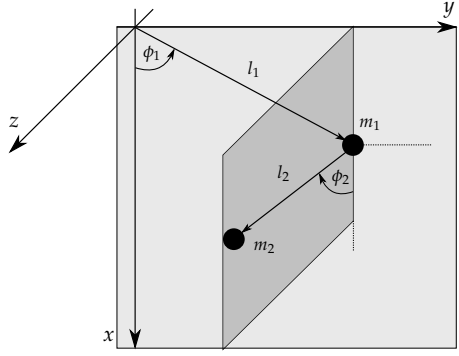}
\caption{Geometry of the first system.}
\label{fig:m_1:geometry}
}
\end{figure}

\begin{equation}
\label{eq:m_1:Lagrangian}
\begin{split}
L=&\frac{1}{2} \left(l_1^2 \left(m_1+m_2\right)
   \dot{\phi }_1^2+l_2^2 m_2 \dot{\phi }_2^2+2 l_1
   l_2 m_2 \dot{\phi }_2 \dot{\phi }_1 \sin
   \phi _1 \sin\phi
   _2\right) \\ &
   +g l_1 m_1 \cos\phi_1+g m_2 \left(l_1
   \cos \phi _1+l_2 \cos\phi
   _2\right).
   \end{split}
\end{equation}
In order to have invertible Legendre transformation we assume that the determinant
of the kinetic energy matrix does not vanish, which is valid provided that
\[
l_1l_2m_2\neq 0.
\]
Then, the Hamiltonian function is the following
\begin{equation}
\label{eq:m_1:Hamiltonian}
\begin{split}
H&=\frac{l_2^2
   m_2 p_1^2 +l_1^2 \left(m_1+m_2\right) p_2^2-2 l_1 l_2 m_2  p_1p_2 \sin \phi _1 \sin \phi _2}{2 l_1^2 l_2^2 m_2 \left(m_2
   \left(1-\sin ^2\phi _1 \sin ^2\phi _2\right)+m_1\right)}\\ &
   -g \left(l_2 m_2 \cos \phi _2+l_1 \left(m_1+m_2\right) \cos \phi
   _1\right),
\end{split}
\end{equation}
and its corresponding canonical equations are given by
\begin{equation}
\label{eq:m_1:vh}
\begin{split}
\dot \phi_1&=\frac{l_2 p_1-l_1 p_2 \sin\phi_1 \sin\phi_2}{l_1^2 l_2
   \left(m_2 \left(1-\sin^2\phi_1 \sin ^2\left(\phi_2\right)\right)+m_1\right)}, \\ \dot \phi_2&=\frac{l_2 m_2 p_1 \sin\phi_1 \sin\phi_2-l_1
   \left(m_1+m_2\right) p_2}{l_1 l_2^2 m_2 \left(m_2 \left(\sin^2\phi_1
   \sin^2\phi_2-1\right)-m_1\right)}, \\
   \dot p_1&=-\frac{m_2 p_1^2 \sin \phi _1 \sin
   ^2\phi _2 \cos \phi _1}{l_1^2 \left(m_2 \left(1-\sin
   ^2\phi _1\sin ^2\phi _2\right)+m_1\right){}^2}-\frac{\left(m_1+m_2\right) p_2^2 \sin \left(2 \phi _1\right) \sin ^2\phi
   _2}{2 l_2^2 \left(m_2 \left(1-\sin^2 \phi_1 \sin ^2\phi
   _2\right)+m_1\right){}^2}\\ &+\frac{p_1
   p_2 \sin\phi_2 \cos \left(\phi _1\right) \left(m_2 \left(\sin
   ^2\phi _1 \sin^2 \phi_2+1\right)+m_1\right)}{l_1 l_2
   \left(m_2 \left(1-\sin^2 \phi_1 \sin ^2\phi
   _2\right)+m_1\right){}^2} -g l_1 \left(m_1+m_2\right) \sin\phi_1,\\ 
\dot p_2&=-\frac{m_2 p_1^2 \sin^2 \phi_1 \sin\phi_2 \cos \phi
   _2}{l_1^2 \left(m_2 \left(1-\sin^2 \phi_1 \sin ^2\phi
   _2\right)+m_1\right){}^2}-\frac{\left(m_1+m_2\right) p_2^2 \sin ^2\phi
   _1 \sin\phi_2 \cos\phi_2}{l_2^2 \left(m_2
   \left(1-\sin^2 \phi_1 \sin ^2\phi
   _2\right)+m_1\right){}^2} \\
   &+\frac{p_1 p_2 \sin\phi_1 \cos
 \phi _2 \left(m_2 \left(\sin^2 \phi_1 \sin ^2\phi
   _2+1\right)+m_1\right)}{l_1 l_2 \left(m_2 \left(1-\sin ^2\phi
   _1 \sin ^2\phi
   _2\right)+m_1\right){}^2}-g l_2 m_2 \sin\phi_2.
\end{split}
\end{equation}

\subsection{Numerical analysis} \begin{figure}[h!]
  \centering \subfigure[$E=-8.9$]{
    \includegraphics[width=0.46\textwidth]{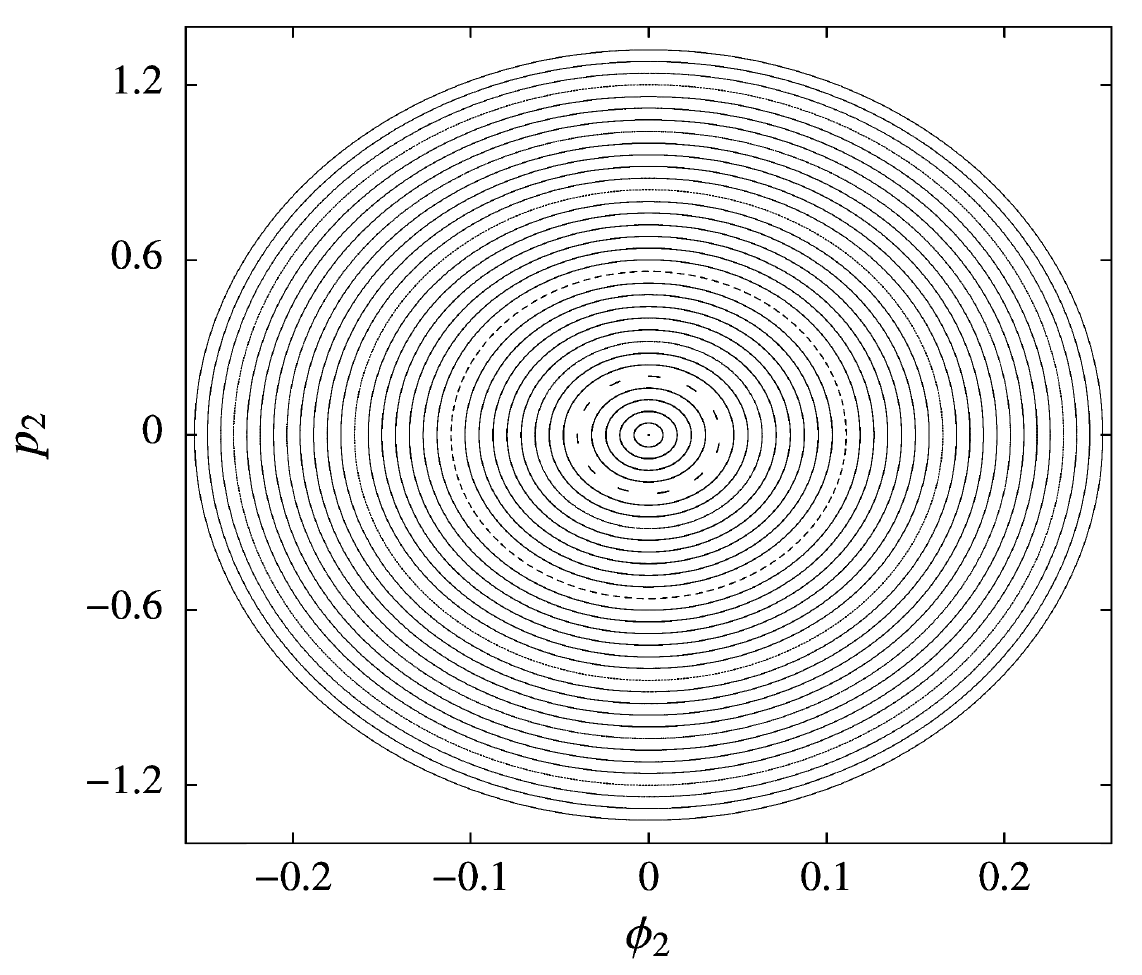}
    \label{subfig:p_m1_a}
  } \subfigure[$E=-6$]{
    \includegraphics[width=0.46\textwidth]{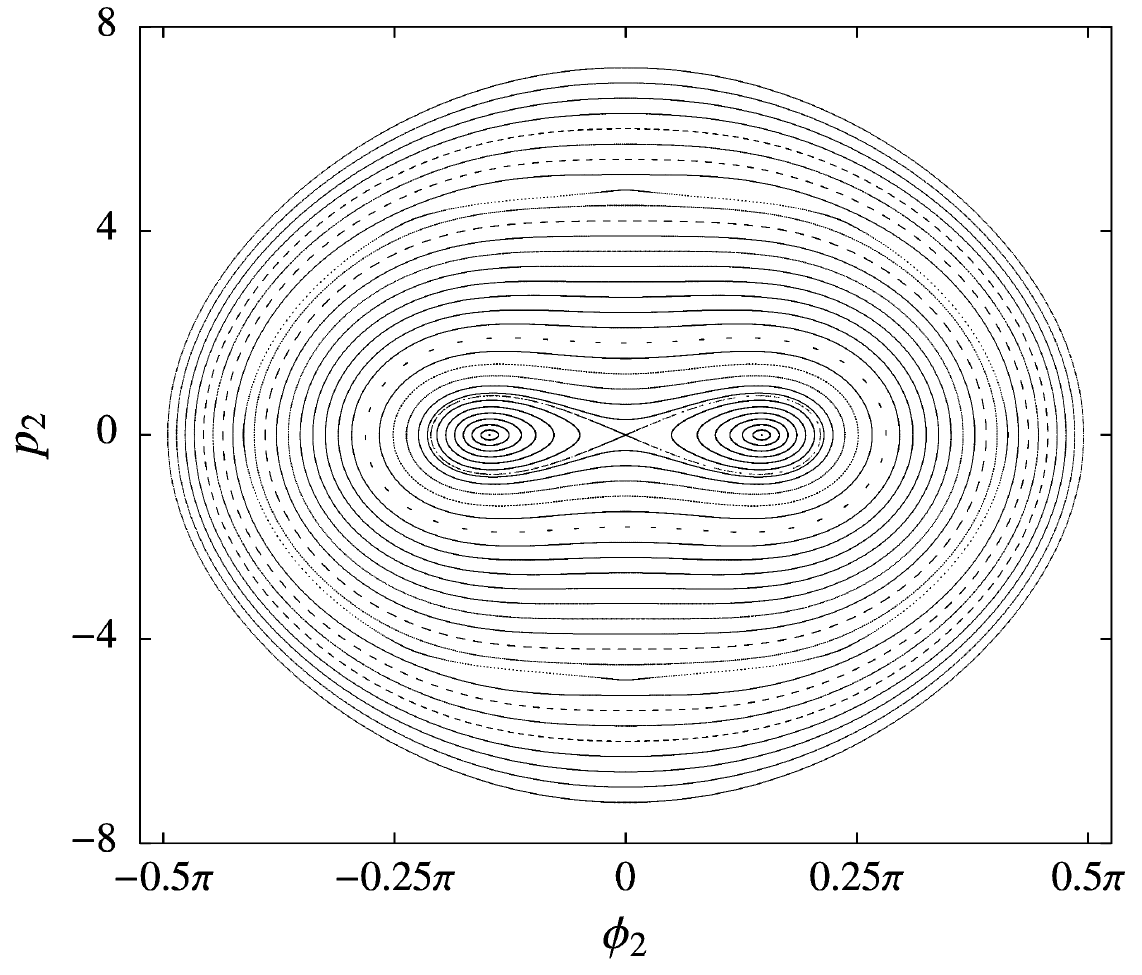}
  }
  \caption{The Poincar\'e cross sections of the first system on the surface
  $\phi_1=0$.\label{fig:p_m1_ab}}
\end{figure}
\begin{figure}[h!]
  \centering \subfigure[$E=-4$]{
    \includegraphics[width=0.46\textwidth]{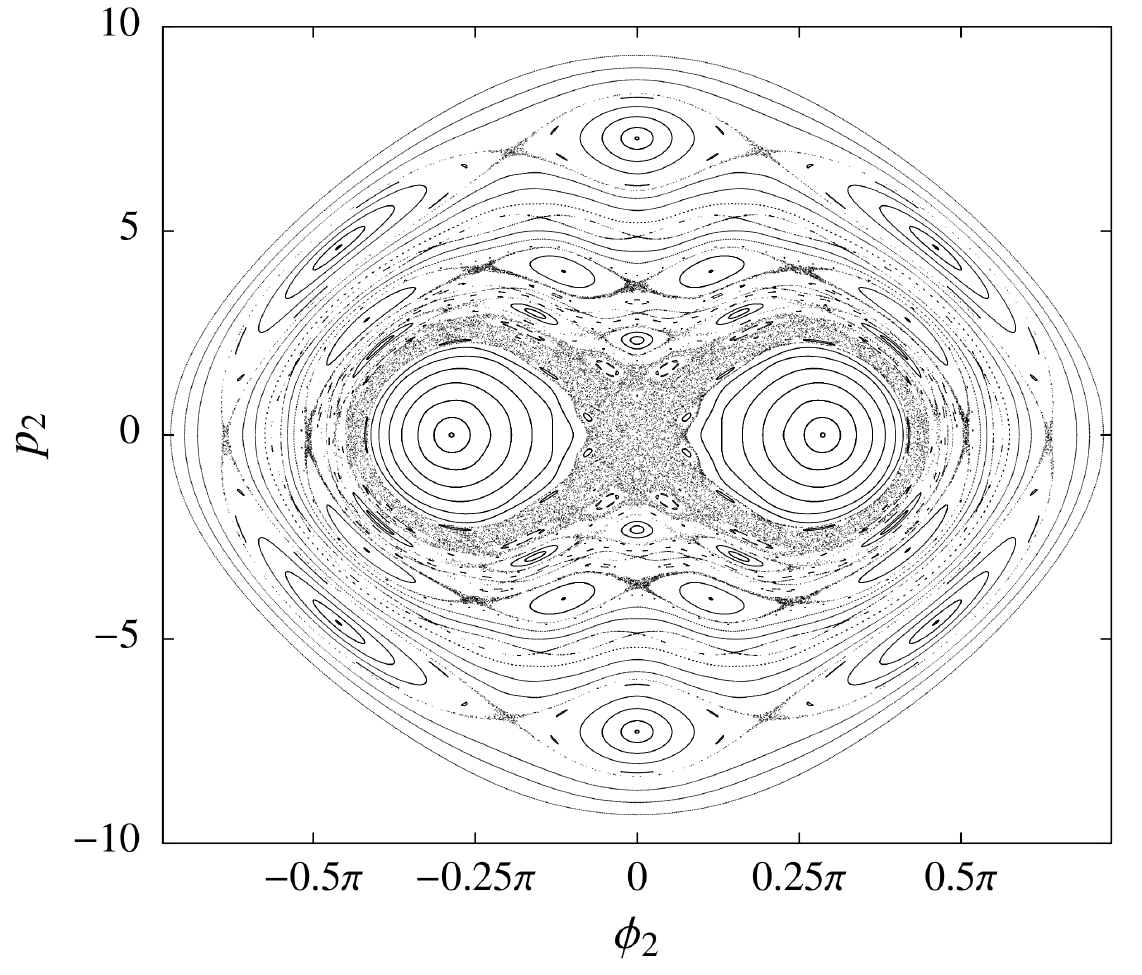}
  } \subfigure[$E=-3$]{
    \includegraphics[width=0.46\textwidth]{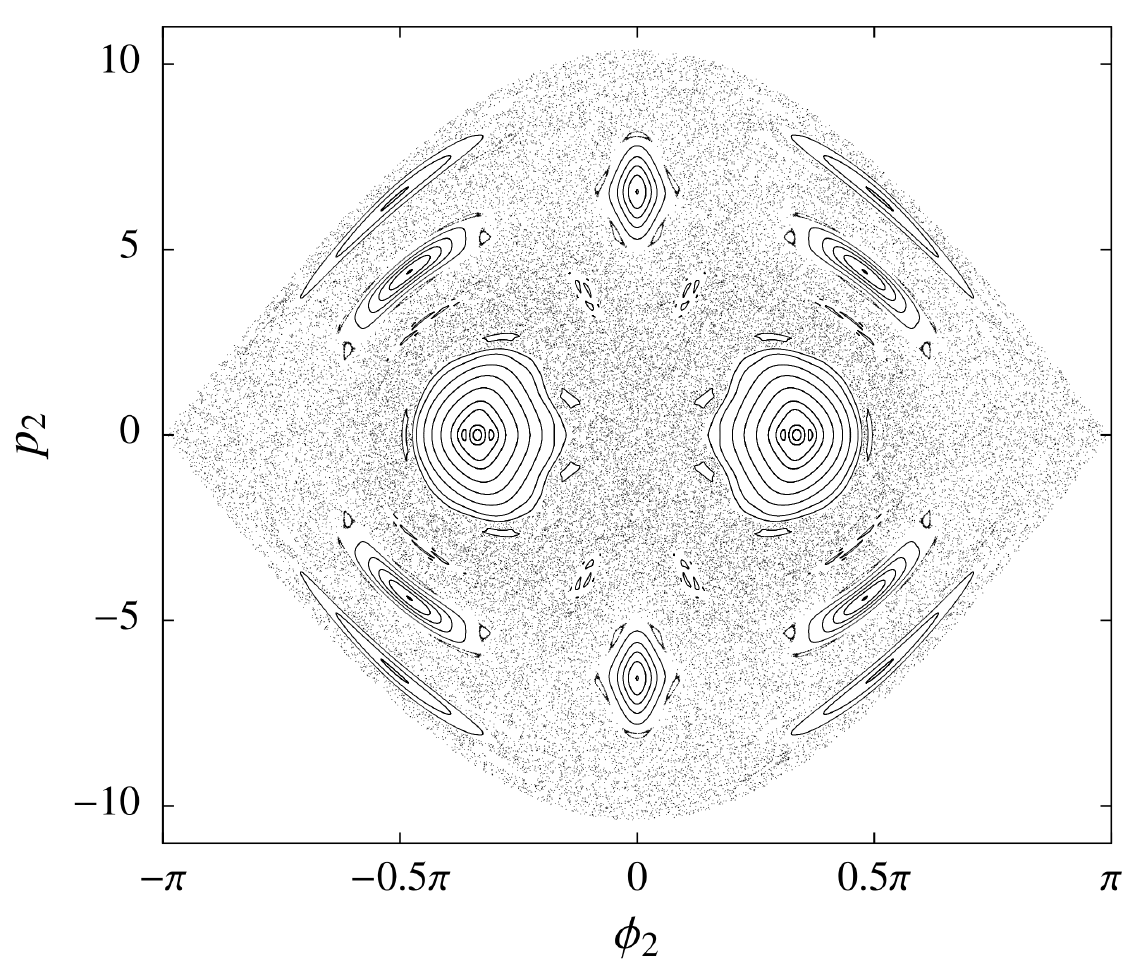}
    \label{fig:p_m1_d}
  }
  \caption{The Poincar\'e cross sections of the first system on the surface
  $\phi_1=0$.\label{fig:p_m1_cd}}
\end{figure}
\begin{figure}[h!]
\begin{center}
\includegraphics[scale=0.46]{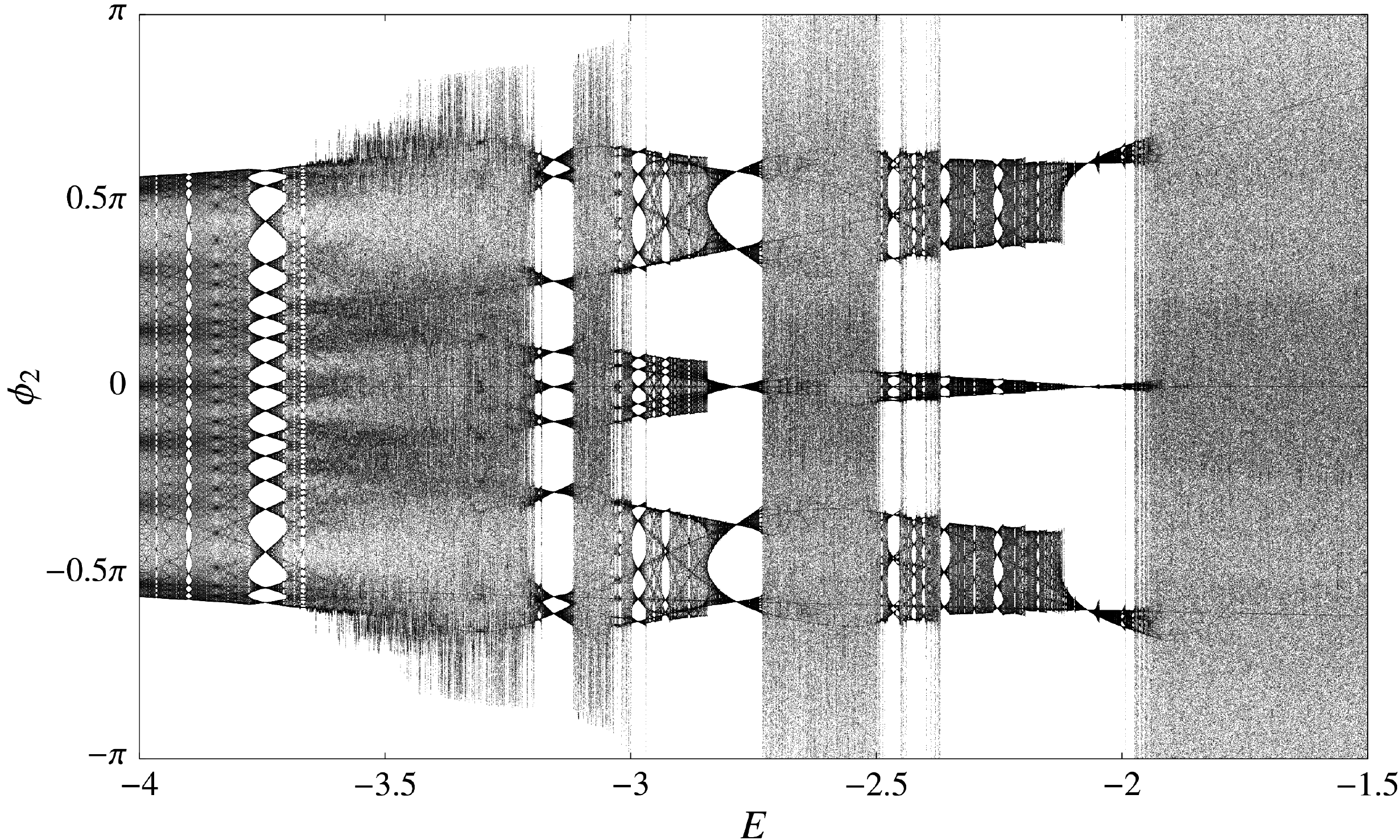}
\caption{The bifurcation diagram with initial conditions $(\phi_2,p_2)=(0,4.9)$.\label{fig:bd_m1_b}}
\end{center}
\end{figure}
As the evolution of our system takes place in a four dimensional phase space,
it is convenient to use the so-called Poincar\'e sections. These are simply the
intersections of orbits with a suitably chosen surface. This and all the other numerical
results were obtained using the Bulirsch-Stoer modified midpoint method with
Richardson extrapolation.
Figures~\ref{fig:p_m1_ab}-\ref{fig:p_m1_cd} present such sections for
increasing values of energy. They were
constructed for the following constant parameters
\begin{equation}
\label{eq:parameters}
m_1=2,\quad m_2=1,\quad l_1=2,\quad l_2=3,\quad g=1,
\end{equation}
with the surface $\phi_1=0$ and $p_1>0$, restricted to the plane
$(\phi_2,p_2)$. It is easy to verify that the global energy minimum
corresponding to the state of rest in the bottom position of the pendula ihas
the value $E_0=-9$.

As expected, at the energy level $E=-8.9$ the pendula oscillate near the
equilibrium point as shown in Figure~\ref{subfig:p_m1_a}. The
image is very regular, in the center we detect the stable particular periodic
solution that is surrounded  by invariant tori.
The situation becomes more complex when we increase energy to the value of
$E=-6$. The invariant tori located at the center become visibly deformed. Some
of them decay and we observe the bifurcations that lead to emergence of stable
periodic solutions which  are enclosed by a separatrix.  Thus, we may expect the
first sign of chaotic behaviour exactly in this region. Our suspicion is
confirmed in Figure~\ref{fig:p_m1_cd}. We can notice at first sight that
the chaotic behaviour appears in the region were the separatrix was located. For
a non-integrable Hamiltonian system, these random-looking points correspond
to the fact that the trajectories can freely wander over larger regions of the
phase space. The trajectories (in contrast to the integrable system) are no
longer confined to the surfaces of a set of nested tori, but they begin to move
outside the tori. Loosely speaking, the tori are destroyed. The
last Poincar\'e section shows the highly chaotic stage of the system. For the
energy $E=-3$ almost all of the regular orbits merge into a global chaotic
region.

Figure~\ref{fig:bd_m1_b} shows the bifurcation diagram representing the
relationship between oscillations in $\phi_2$ and energy $E$. For a given initial
condition $(\phi_2,p_2)=(0,4.9)$ we effectively construct the Poincar\'e cross
sections with cross plane $\phi_1=0$ with gradually increasing energy. The
periodic motion diverges as the energy increases and finally becomes chaotic.
We can also detect the stable ``windows'' between completely chaotic regions.

Clearly, the numerical analysis demonstrates that the system is chaotic, and the
main goal of this article is to also prove that the system is not integrable
for a wide range of parameters (almost everywhere).

\subsection{The non-integrability proof}
Below we formulate the main theorem of this subsection
\begin{theorem}
The system governed by Hamiltonian~\eqref{eq:m_1:Hamiltonian} is not integrable
in the Liouville sense in the class of meromorphic functions of coordinates and
momenta, except possibly for energy
$E\in\{-g[l_1(m_1+m_2)+l_2m_2],-g[l_1(m_1+m_2)-l_2m_2]\}$.
\end{theorem}
\begin{proof}
The system~\eqref{eq:m_1:vh} has two known invariant manifolds
\[
\scN_1=\left\{(\phi_1,\phi_2,p_1,p_2)\in \C^4\ |\ \phi_1=p_1=0\right\},\quad \scN_2=\left\{(\phi_1,\phi_2,p_1,p_2)\in \C^4\ |\ \phi_2=p_2=0\right\}.
\]
For further analysis we chose $\scN_1$, due to considerably simpler
characteristic exponents, which will become apparent in a moment. Restricting
the right hand sides of~\eqref{eq:m_1:vh} to $\scN_1$, we obtain 
\begin{equation}
\label{eq:m_1:vh_0}
\dot \phi_1=0, \quad
\dot \phi_2=\frac{p_2}{l_2^2m_2},\quad \dot p_1=0,\quad \dot p_2=-gl_2m_2\sin\phi_2.
\end{equation}
Hence, we have a one-parameter family of particular solutions lying on $H=E$
\begin{equation}
\label{eq:m_1:H_0}
\dot \phi_2^2=\frac{2 \left[E+g l_2 m_2 \cos \phi
   _2+g l_1
   \left(m_1+m_2\right)\right]}{l_2^2 m_2}.
\end{equation}
Solving equations~\eqref{eq:m_1:vh_0} and taking into
account~\eqref{eq:m_1:H_0}, we obtain our particular solution 
$\vvarphi(t)=(0,\phi_2(t),0,p_2(t))$, where by a slight abuse of notation we use
the variable symbols for particular functions of time.
Let $\boldsymbol{v}=(\Phi_1,\Phi_2,P_1,P_2)^T$ denote the variation of
$(\phi_1,\phi_2,p_1,p_2)^T$, then the variational equations along $\vvarphi(t)$
take the form
\begin{equation}
\label{eq:m_1:var}
\begin{split}
    &\frac{\mathrm{d}\boldsymbol{v}}{\mathrm{d}t} 
    = \boldsymbol{L}\boldsymbol{v},\\
    &\boldsymbol{L}=\begin{pmatrix}
 -\frac{p_2\sin \phi _2 }{l_1 l_2
   m_3} & 0 & \frac{1}{l_1^2
   m_3} & 0 \\
 0 & 0 & 0 & \frac{1}{l_2^2 m_2} \\
 -\frac{p_2^2\sin ^2\phi _2 }{l_2^2
   m_3}-g l_1 m_3
   & 0 & \frac{p_2\sin \phi _2 }{l_1
   l_2 m_3} & 0 \\
 0 & -g l_2 m_2 \cos\phi _2 & 0 & 0
   \\
\end{pmatrix},
\end{split}
\end{equation}
where $m_3=m_1+m_2$.
Notice that equations for $\Phi_1$ and $P_1$ form a closed subsystem, called the normal
variational equations, that can be rewritten as one second order differential
equation
\begin{equation}
\label{eq:m_1:ddf1}
\ddot \Phi+\frac{p_2^2 \cos \phi
   _2+g l_2^3 m_2 \left(m_2 \cos ^2\phi
   _2+m_1\right)}{l_1 l_2^3 m_2 \left(m_1+m_2\right)}\Phi=0,\qquad \Phi=\Phi_1.
\end{equation}
Next, by means of the change of the independent variable 
\begin{equation}
\label{eq:m_1:rationalization}
t\longrightarrow z=\cos \phi_2(t),
\end{equation} and transformation of derivatives
\[
\Dt=\dot z\Dz,\quad \Dtt=\dot z^2\Dzz+\ddot z\Dz,\]
we can rewrite equation~\eqref{eq:m_1:ddf1} as
\begin{equation}
\label{eq:m_1:rational}
\Phi''+p(z)\Phi'+q(z)\Phi=0,\quad '\equiv \frac{\mathrm{d}}{\mathrm{d}z},
\end{equation}
with rational coefficients
\begin{equation}
p = \frac{z}{z^2-1}+\frac{1}{2 (z+\gamma)},\quad
q = -\frac{\beta +3 z^2+2 \gamma  z}{2 \alpha  (\beta
   +1) \left(z^2-1\right) (z+\gamma )},
\end{equation} 
where in the last steep
we used~\eqref{eq:m_1:H_0}. The explicit forms of the parameters
$(\alpha,\beta,\gamma)$ are the following
\[
\alpha=\frac{l_1}{l_2},\qquad \beta=\frac{m_1}{m_2},\qquad \gamma=\alpha+\alpha\beta+\frac{E}{l_2m_2g}.
\]
Next, let us apply the classical change of the dependent variable
\begin{equation}
\label{eq:m_1:Tchihandrius}
\Phi=w\exp\left[-\frac{1}{2}\int_{z_0}^zp(s)\rmd s\right],
\end{equation}
which transforms~\eqref{eq:m_1:rational} into its reduced form
\begin{equation}
\label{eq:m_1:normal}
w''=r(z)w,\qquad r(z)=-q(z)+\frac{1}{2}p'(z)+\frac{1}{4}p(z)^2,
\end{equation}
where the coefficient $r(z)$ is given by
\[
r(z)=-\frac{3}{16}\left[\frac{1}{(z-1)^2}+\frac{1}{(z+1)^2}+\frac{1}{(z+\gamma)^2}\right]+\frac{4 \beta +\gamma  (\alpha  \beta +\alpha +8
   z)+3 z (\alpha  \beta +\alpha +4 z)}{8 \alpha 
   \left(z^2-1\right) (z+\gamma )(\beta +1) }.
\]
In order to avoid the confluence of the singularities we assume that $\gamma\neq \pm
1$. This corresponds to the  fact that certain values of energy should be
excluded, namely
\begin{equation}
\label{eq:m_1:energy_exclude}
\{E_1,E_2\}= \left\{-g[l_1(m_1+m_2)+ l_2m_2],-g[l_1(m_1+m_2)- l_2m_2]\right\}.
\end{equation}
Notice that these two energies are the stationary point energies of the
system~\eqref{eq:m_1:Hamiltonian}. The first one, $E_1$, is the global energy
minimum corresponding to the equilibrium solution when both pendula are at
rest. The second energy, $E_2$, corresponds to the case when the first pendulum
is pointing down and the second one up, both at rest, although the equilibrium
position is not the only solution with this energy.

If there existed an additional first
integral, it would not depend on the energy value, so in particular it would
exist for all generic values of energy other than $E_1$ and $E_2$. In order to
exclude such an integral we can thus safely assume
that~\eqref{eq:m_1:energy_exclude} is always satisfied.

However, it could happen that there is a first integral, which is conserved
only on those special hypersurfaces. We cannot preclude its existence from the
below analysis, but note that such an isolated constant of motion would be of
much less practical value. Also, the physical solution corresponding to $E_1$
is just the equilibrium position; while the Poincar\'e section in
Figure~\ref{fig:p_m1_d} shows heavy chaos for energy $E_2$ indicating that not
even a local first integral exists.

It is important to note that transformations~\eqref{eq:m_1:rationalization}
and~\eqref{eq:m_1:Tchihandrius} change the differential Galois group of
equation~\eqref{eq:m_1:ddf1}. However, the key is that they do not change the
identity component of this group, see~\cite{Morales:99::}. Hence, in order to
make the non-integrability proof it is enough to show that the identity
component of the differential Galois group of~\eqref{eq:m_1:normal} is not
Abelian.
For equation~\eqref{eq:m_1:normal} the differential Galois group $\scG$ is an
algebraic subgroup of $\operatorname{SL}(2,\C)$. The following lemma describes
all possible types of $\scG$ and relates them to the forms of a solution
of~\eqref{eq:m_1:normal}, see \cite{Kovacic:86::,Morales:99::}.
\begin{lemma}
\label{lem:m_1_a}
Let $\scG$ be the differential Galois group of equation~\eqref{eq:m_1:normal}.
Then one of the four cases can occur.
\begin{enumerate}
\item $\scG$ is conjugate to a subgroup of triangular group
\[
\scT=\left\{\begin{pmatrix}
a&b\\
0 &a^{-1}
\end{pmatrix}\vert a\in \C^*, b\in \C\right\},
\]
and equation~\eqref{eq:m_1:normal} has an exponential solution  
$w=P\exp[\int \omega]$, $P\in \C[z],\ \omega \in \C(z)$.
\item $\scG$ is conjugated with a subgroup of
\[
\scD^\dagger=\left\{\begin{pmatrix}
c&0\\ 0&c^{-1}
\end{pmatrix}\vert c\in \C^*\right\} \cup 
\left\{\begin{pmatrix}
0&c\\ -c^{-1}&0
\end{pmatrix}\vert c\in \C^*\right\};
\]
in this case~\eqref{eq:m_1:normal} has a solution of the form 
$w=\exp[\int \omega]$, where $\omega$ is algebraic function of degree $2$.
\item $\scG$ is  finite and all solutions of equation~\eqref{eq:m_1:normal} are algebraic.
\item $\scG=\operatorname{SL}(2,\C)$ and equation~\eqref{eq:m_1:normal} has no Liouvillian solution.
\end{enumerate}
\end{lemma}
\begin{remark}
Let us write $r(z)\in \C$ in the form
\[
r(z)=\frac{s(z)}{t(z)},\qquad s(z),t(z)\in \C[z].
\]
The roots of $t(z)$ are the poles of $r(z)$. The order of the pole is the
multiplicity of the zero of $t$, and  the order of $r(z)$ at $\infty$ is
$\operatorname{deg}(t)-\operatorname{deg(s)}$.
\end{remark}
\begin{lemma}
\label{lem_m_1_b}
The following conditions are necessary for the respective cases given in Lemma~\ref{lem:m_1_a}.
\begin{enumerate}
\item Every pole of $r$ must have even order or else have order $1$. The order
    of $r$ at $\infty$ must be even or else grater than $2$.
\item $r$ must have at least one pole that either has odd order greater than
    $2$ or else has order $2$.
\item The order of a pole of $r$ cannot exceed $2$ and the order of $r$ at
    $\infty$ must be at least $2$. If the partial fraction expansion of $r$ is
\[
r(z)=\sum_i\frac{a_i}{(z-z_i)^2}+\sum_j\frac{b_j}{z-z_j},
\]
then $\Delta_i=\sqrt{1+4a_i}\in \Q$ for each $i$, $\sum_j b_j=0$ and if
\[
g=\sum_ia_i+\sum_j b_jd_j,
\]
then $\sqrt{1+4g}\in \Q$.
\end{enumerate}
\end{lemma}
Equation~\eqref{eq:m_1:normal} has four singularities
\[
z_{1,2}=\pm 1,\quad z_3=-\gamma,\quad z_\infty=\infty.
\] 
For $\gamma\neq \pm 1$ the singularities $\{z_1,z_2,z_3\}$ are poles of the
second order, and the degree of infinity is $1$.
Taking into account the character of singularities we deduce, according to the
Lemma~\ref{lem_m_1_b}, that the differential Galois group of reduced
equation~\eqref{eq:m_1:normal} cannot be reducible or finite. Differential
Galois group can be only dihedral or $\operatorname{SL}(2,\C)$.  In order to
check the first possibility we apply the second case of the so-called Kovacic
algorithm.  Here and below we will use the original formulation of this algorithm
given in~\cite{Kovacic:86::} as well as the notation introduced in this paper.
\begin{lemma}
The differential Galois group of equation~\eqref{eq:m_1:normal} is $\operatorname{SL}(2,\C)$.
\end{lemma}
\begin{proof}
 First, for singular points $z_i$, we calculate the auxiliary sets
\[
E_i=\left\{2+k\Delta_i\ |\ k=0,\pm 2\right\} \cap \Z, \quad i=1,2,3,
\]
where $\Delta_i$ are the difference of exponents defined by
$\Delta_i=\sqrt{1+4a_i}$, and $a_i$ are the coefficients appearing in expansion
of $r(z)$
\[
r=\sum_{i=1}^3\frac{a_i}{(z-z_i)^2}.
\] 
In our case $a_i=-3/16$. Since the order of infinity is one, the auxiliary  sets are
\begin{equation}
E_1=E_2=E_3=\{1,2,3\},\qquad E_\infty =\{1\}.
\end{equation}
Next, following the algorithm we look for elements $e=(e_\infty,e_1,e_2,e_3)$
of Cartesian product $E=E_\infty\times \prod_{i=1}^3 E_i$ such that
\begin{equation}
\label{eq:m_1_dn}
d(e):=\frac{1}{2}\left(e_\infty-\sum_{i=1}^3e_i\right)\in \N_0,
\end{equation}
where $\N_0$ is a set of non-negative integers with $0$. In is not difficult to
note that there are no elements of $E$ satisfying this condition. Thus,
according to the algorithm the second case cannot occur, which implies that
only the fourth case is possible, i.e., $G=\operatorname{SL}(2,\C)$, and
equation~\eqref{eq:m_1:normal} has no Liouvillian
solution. 
\renewcommand{\qedsymbol}{$\blacksquare$}
\end{proof}
From our short analysis we can formulate the following conclusion. Since the
differential Galois group of reduced equation~\eqref{eq:m_1:normal} is
$\operatorname{SL}(2,\C)$, the identity component of differential Galois group
of variational equations~\eqref{eq:m_1:var} is not Abelian. This implies that
the model of two pendula  restricted to $(x,y)$ and $(x,z)$ planes
respectively, is not integrable in the class of function meromorphic in
coordinates and momenta.
\end{proof}

\section{Model 2: toroidal pendulum\label{sec:model_2}}
Let us consider the the second model presented in
Figure~\ref{fig:m_2:geometry}. Like in the previous case, the first pendulum
with mass $m_1$ and length $l_1$ is attached to the fixed point and moves in
the $(x,y)$ plane. However, the second one with parameters $m_2$ and $l_2$ 
is restricted to a plane containing the first pendulum and perpendicular to the
$(x,y)$ plane. In other words it moves on a torus. The positions of the masses
$m_1$ and $m_2$, can be defined parametrically in the following way
\[
 \vl_1=l_1[\cos\phi_1,\sin\phi_1,0]^T,\quad \vl_2=\left[(l_1+l_2\cos\phi_2)\cos\phi_1,(l_1+l_2\cos\phi_2)\sin\phi_1,l_2\sin\phi_2\right]^T.
\]
\begin{figure}[h!]
\centering{
\resizebox{77mm}{!}{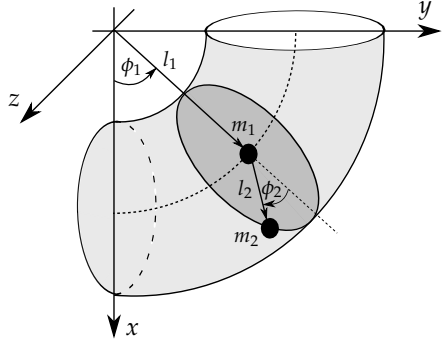}
\caption{Geometry of the second system.}
\label{fig:m_2:geometry}
}
\end{figure}
From this, we can write immediately the Lagrange function
\begin{equation}
\label{eq:m2:lagrange}
L=\frac{1}{2}\left([l_1^2m_1+m_2(l_1+l_2\cos\phi_2)^2]\dot\phi_1^2+l_2m_2\dot\phi_2^2\right)+
g[l_1(m_1+m_2)+l_2m_2\cos\phi_2]\cos\phi_1.
\end{equation}
This form is quadratic in velocities and is non-singular provided $l_2m_2\neq
0$ and then the Legendre transformation can be carried out. From now on we
assume that this condition is always fulfilled. The Hamiltonian function is
given by
\begin{equation}
\label{eq:m_2:hamiltonian}
H=\frac{1}{2}\left(\frac{p_1^2}{l_1^2m_1+m_2(l_1+l_2\cos\phi_2){}^2}+\frac{p_2^2}{l_2^2m_2}\right)-g[l_1(m_1+m_2)+l_2m_2\cos\phi_2]\cos\phi_1,
\end{equation}
and the equations of motion are as follows
\begin{equation}
\label{eq:m_2:vh}
\begin{split}
\dot \phi_1&=\frac{p_1}{l_1^2m_1+m_2(l_1+l_2\cos\phi_2){}^2},\quad \dot p_1=-g[l_1(m_1+m_2)+l_2m_2\cos\phi_2]\sin\phi_1, \\
\dot \phi_2 &=\frac{p_2}{l_2^2m_2},\quad \dot p_2=-l_2m_2\left(\frac{(l_1+l_2\cos\phi_2)p_1^2}{[l_1^2m_1+m_2(l_1+l_2\cos\phi_2)^2]^2}+g\cos\phi_1\right)\sin\phi_2.
\end{split}
\end{equation}

\subsection{Numerical analysis}

We analyse the dynamics of this model with  the same values of parameters as in
the first one, see~\eqref{eq:parameters}.
Figures~\ref{fig:p_m2_ab}-\ref{fig:p_m2_cd} present the Poincar\'e cross
sections on the $(\phi_2,p_2)$ plane with $\phi_1=0$ and $p_1>0$. As
previously,
the energy minimum corresponding to equilibria of both pendula is $E_0=-9$.
We note that for the energy  level $E=-8.9$ the first section presented
in~\ref{subfig:p_m2_a} is strictly similar to the one (see
Fig.~\ref{subfig:p_m1_a}) given in our first model. The whole
figure is filled with quasi-periodic orbits surrounding the stable periodic
particular solution. Looking at Figure~\ref{subfig:p_m2_a} we can also
detect the region of stable periodic solution related to high order resonance.
The situation becomes more complex when we increase energy to the value of
$E=-5$ and the invariant tori deform strongly.
Some of them decay giving rise to stable periodic solutions enclosed by
separatrices. Thus, as expected for higher values of energies the invariant
tori become more and more deformed and the region corresponding to chaotic
motion appears, see Figure~\ref{fig:p_m2_cd}. Figure~\ref{fig:bd_m2_b}
shows the bifurcation diagram for the coordinate $\phi_2$ as a function of $E$.
The successive sections were constructed by choosing the cross plane $\phi_1=0$
with the the initial condition $(\phi_2,p_2)=(0.61,0)$.
\begin{figure}[h!]
  \centering \subfigure[$E=-8.9$]{
    \includegraphics[width=0.46\textwidth]{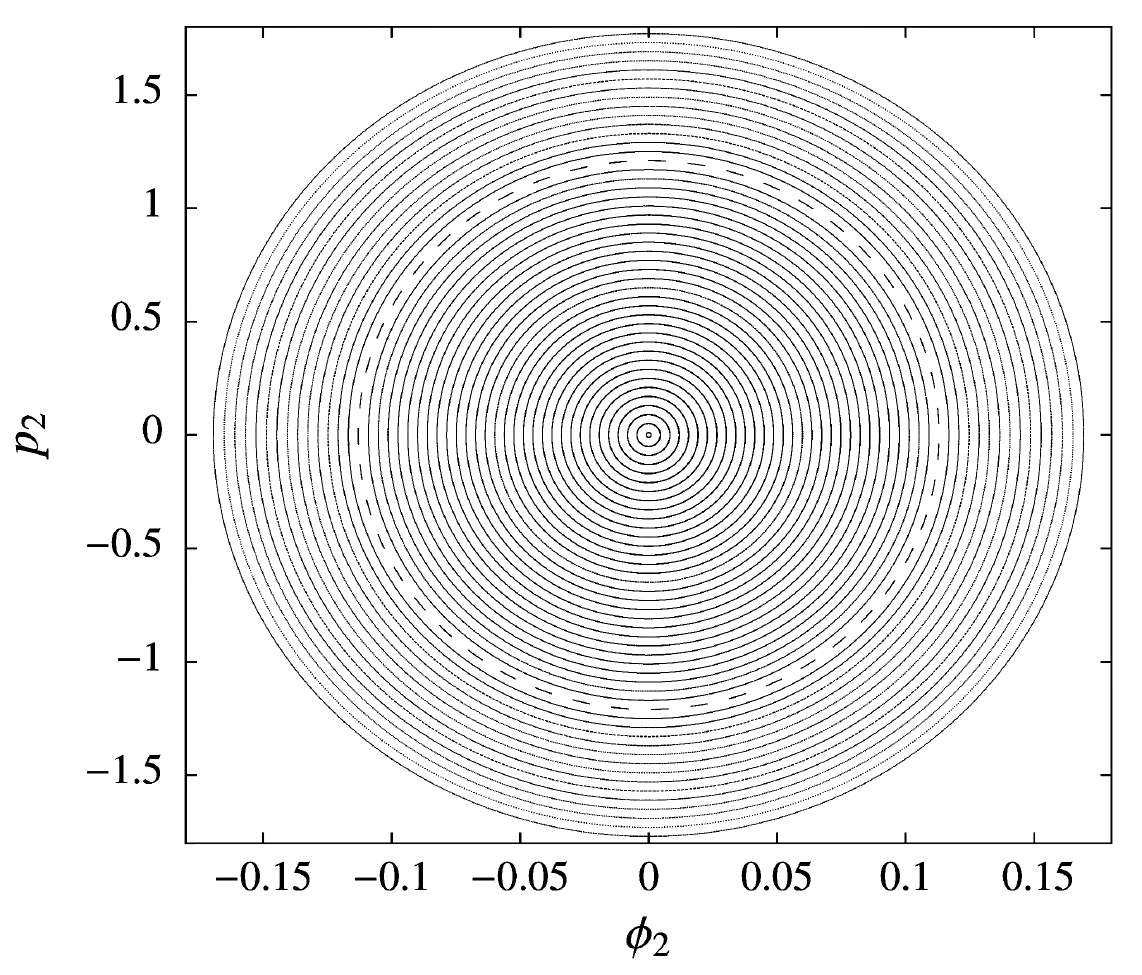}
    \label{subfig:p_m2_a}
  } \subfigure[$E=-5$]{
    \includegraphics[width=0.46\textwidth]{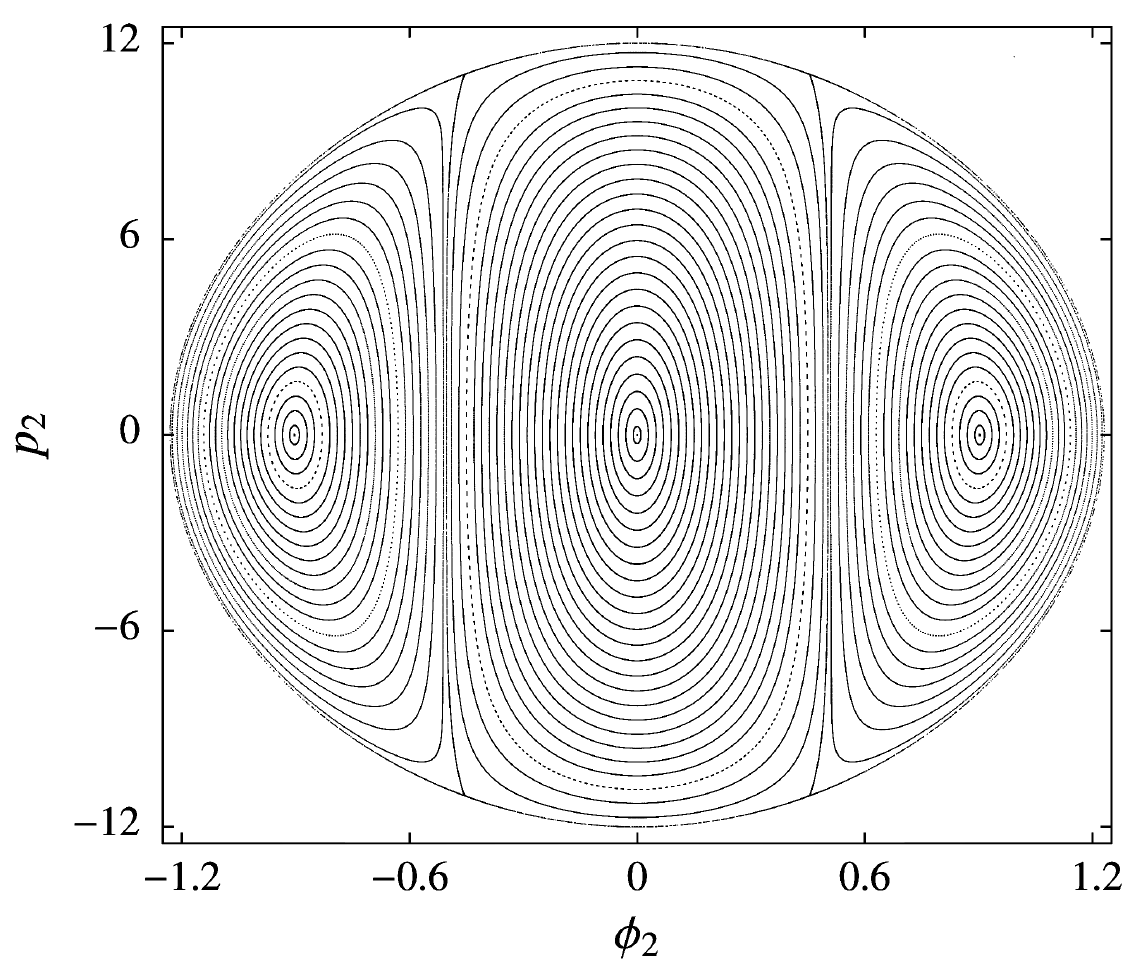}
  }
  \caption{The Poincar\'e cross sections of the second system on the surface
  $\phi_1=0$.\label{fig:p_m2_ab}}
\end{figure}
\begin{figure}[h!]
  \centering \subfigure[$E=-2.2$]{
    \includegraphics[width=0.46\textwidth]{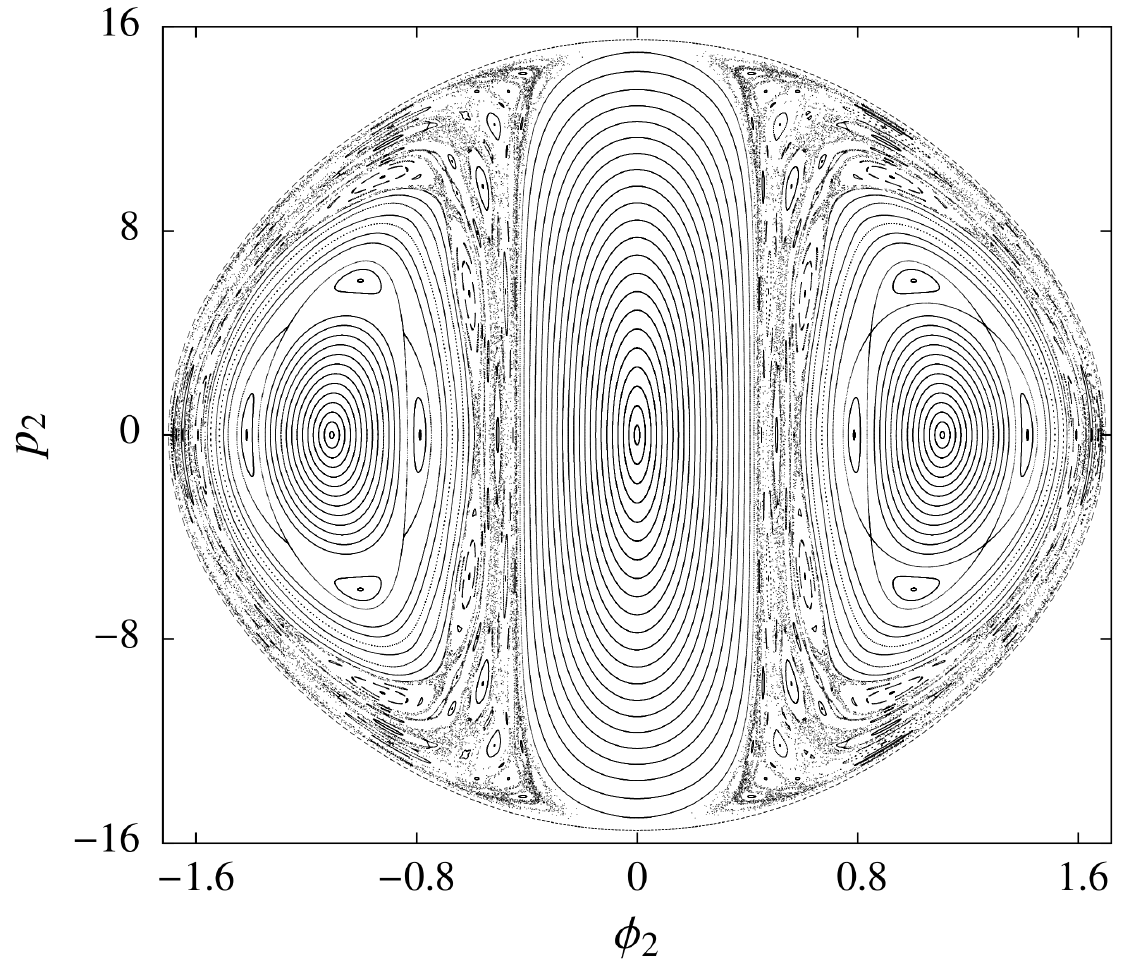}
  } \subfigure[$E=0$]{
    \includegraphics[width=0.46\textwidth]{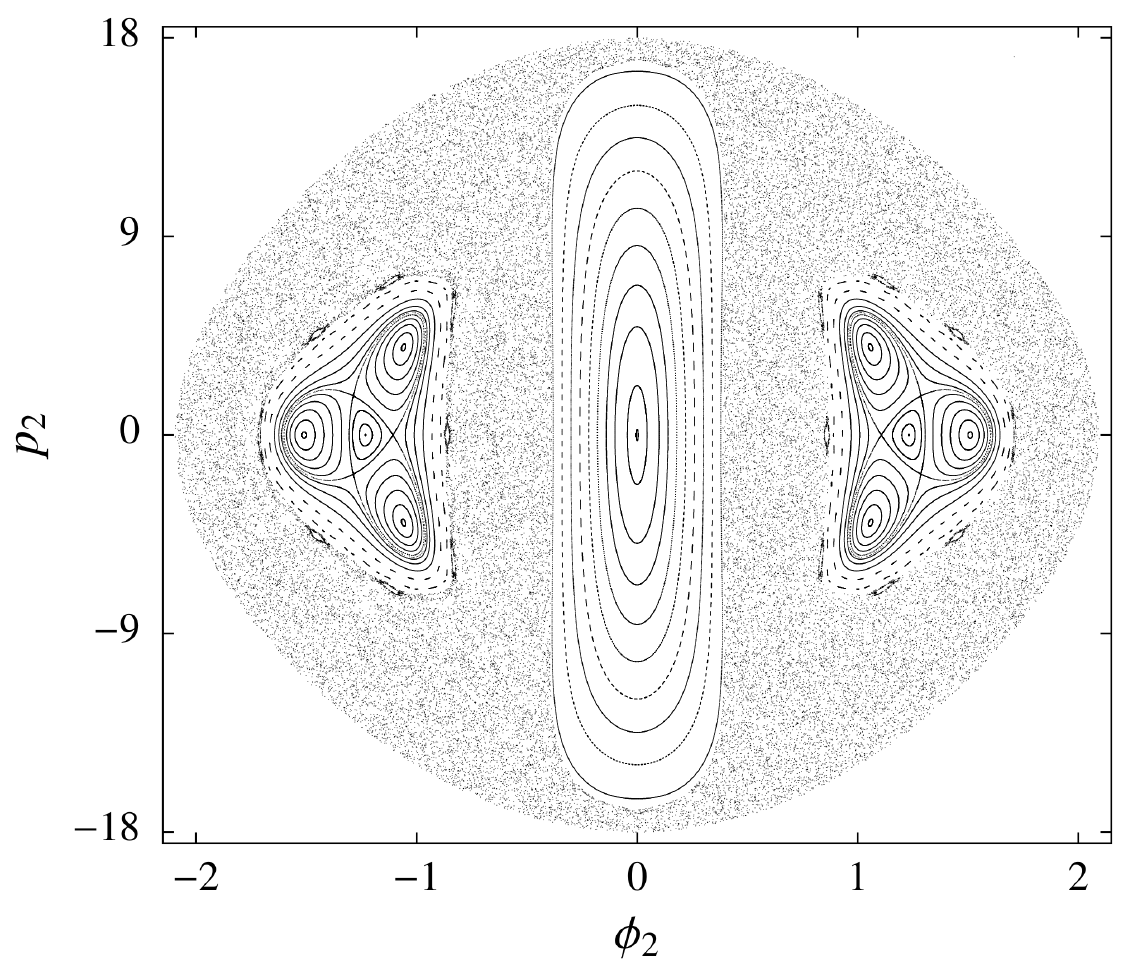}
  }
  \caption{The Poincar\'e cross sections of the second system on the surface
  $\phi_1=0$.\label{fig:p_m2_cd}}
\end{figure}
\begin{figure}[h!]
\begin{center}
\includegraphics[scale=0.46]{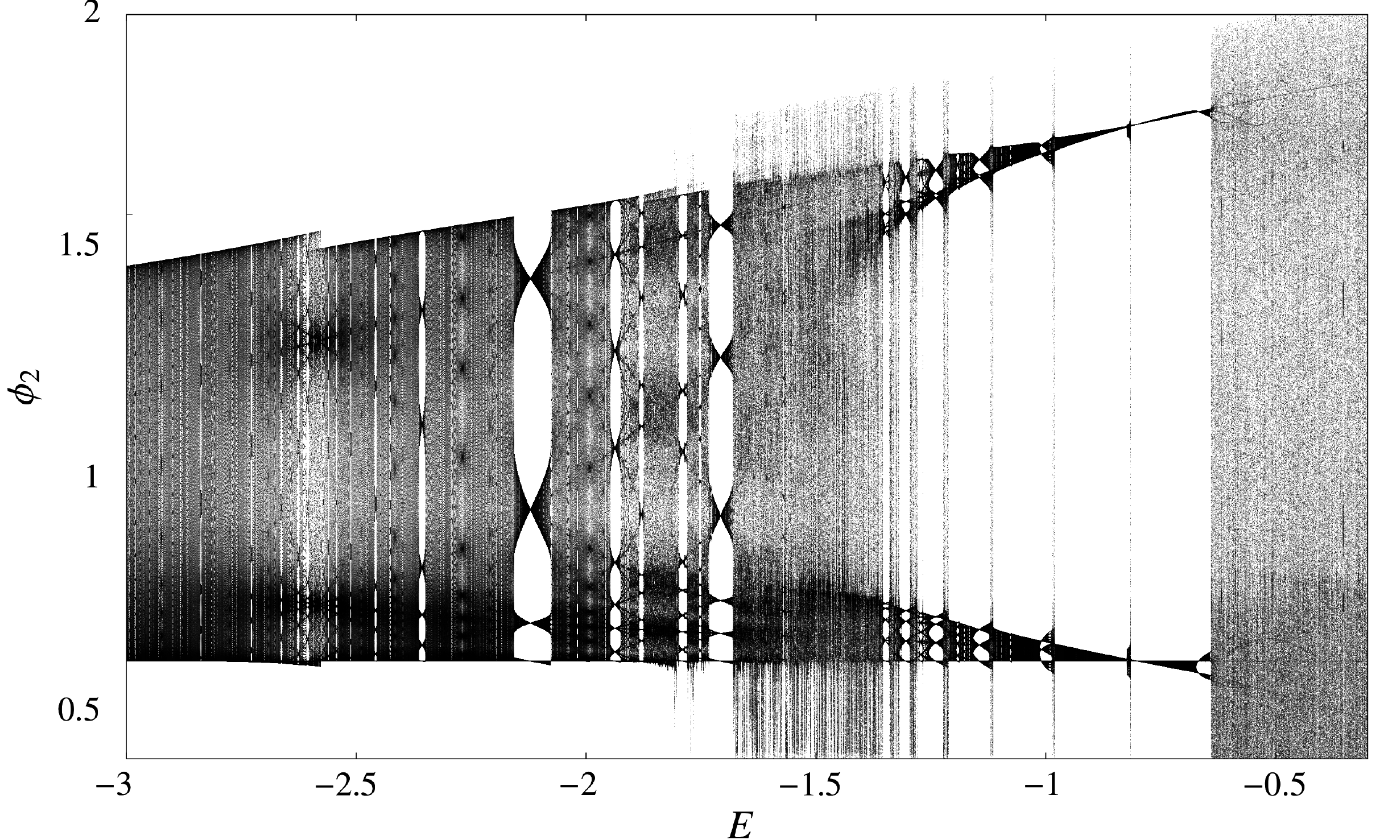}
\caption{The bifurcation diagram with initial conditions $(\phi_2,p_2)=(0.61,0)$.\label{fig:bd_m2_b}}
\end{center}
\end{figure}
\subsection{The non-integrability proof}
Let us formulate the main theorem of this subsection.
\begin{theorem}
The system governed by Hamiltonian~\eqref{eq:m_2:hamiltonian} has no
additional meromorphic first integral for 
\begin{equation}
E\notin \left\{-g(l_1(m_1+m_2)\pm l_2m_2), -gl_1(m_1\pm \rmi
\sqrt{\frac{m_1}{m_2}}m_2)\right\}.
\end{equation}
In other words, it is not integrable in the Liouville sense in the class of
meromorphic functions of coordinates and momenta.
\end{theorem}
\begin{proof}
The system of equations~\eqref{eq:m_2:vh} has two known invariant manifolds defined by
\[
\scN_1=\left\{(\phi_1,\phi_2,p_1,p_2)\in \C^4\ |\ \phi_1=p_1=0\right\},\quad \scN_2=\left\{(\phi_1,\phi_2,p_1,p_2)\in \C^4\ |\ \phi_2=p_2=0\right\}.
\]
Indeed, restricting the equations~\eqref{eq:m_2:vh} to $\scN_1$, we obtain
\begin{equation}
\label{eq:m_2:vh0}
\dot \phi_1=0,\quad \dot \phi_2=\frac{p_2}{l_2^2m_2},\quad \dot p_1=0,\quad \dot p_2=-gl_2m_2\sin\phi_2.
\end{equation}
Thus, we have a one-parameter family of particular solutions lying on $H=E$
\begin{equation}
\label{eq:m_2:family}
\dot \phi_2^2=\frac{2 \left[E+g l_2 m_2 \cos \phi
   _2+g l_1
   \left(m_1+m_2\right)\right]}{l_2^2 m_2}.
\end{equation}
If we denote by $[\Phi_1,\Phi_2,P_1,P_2]^T$ the variations of
$[\phi_1,\phi_2,p_1,p_2]^T$, then the variational equations along this
particular solutions are the following
\begin{equation}
\begin{pmatrix}
\dot \Phi_1\\
\dot \Phi_2\\
\dot P_1\\
\dot P_2
\end{pmatrix}=
\begin{pmatrix}
0&0&a_{13}&0\\
0&0&0&a_{24}\\
a_{31}&0&0&0\\
0&a_{42}&0&0 \\
\end{pmatrix}
\begin{pmatrix}
\Phi_2\\ \Phi_2 \\ P_1 \\ P_2
\end{pmatrix},
\end{equation}
with
\begin{equation}
\begin{aligned}
a_{13}&=\frac{1}{l_1^2m_1+m_2(l_1+l_2\cos\phi_2)^2}, & a_{24} &=\frac{1}{l_2^2 m_2},\\
a_{31}&=-g(l_1(m_1+m_2)+l_2m_2\cos\phi_2, & a_{42} &= -gl_2m_2\cos\phi_2.
\end{aligned}
\end{equation}
The system for $\dot \Phi_1$, and $P_1$ forms normal variational equations
that can be rewritten  as a one second-order differential equation
\begin{equation}
\label{eq:m_2:variational}
\ddot \Phi+a\dot \Phi+b\Phi=0,\qquad \Phi\equiv \Phi_1,
\end{equation}
with coefficients
\[
a=-\frac{\dot a_{13}}{a_{13}}=-\frac{2(l_1+l_2\cos\phi_2)p_2\sin\phi_2}{l_2(l_1^2m_1+m_2(l_1+l_2\cos\phi_2)^2},\quad b=-a_{13}a_{31}=\frac{g[l_1(m_1+m_2)+l_2m_2\cos\phi_2]}{l_1^2m_1+m_2(l_1+l_2\cos\phi_2)^2}.
\]
In order to transform this equation into one with rational coefficients we
make the following change of the independent variable
\begin{equation}
t\longrightarrow z=\cos\phi_2(t).
\end{equation}
After this transformation the variational equation~\eqref{eq:m_2:variational}
converts into
\begin{equation}
\label{eq:m_2:rational}
\Phi''+p(z)\phi'+q(z)\Phi=0,\quad '\equiv \Dz,
\end{equation}
with coefficients
\[
    p=\frac{z}{z^2-1}+\frac{2(z+\alpha)}{(z+\alpha)^2+\alpha^2\beta}+\frac{1}{2(z+\gamma)},\quad q=-\frac{\alpha  \beta +\alpha +z}{2
   \left(z^2-1\right) (\gamma +z) \left(\alpha ^2
   \beta +(\alpha +z)^2\right)},
\]
where 
\begin{equation}
\label{eq:m_2:parameters}
\alpha=\frac{l_1}{l_2},\qquad \beta=\frac{m_1}{m_2},\qquad
\gamma=\frac{E+gl_1(m_1+m_2)}{gl_2m_2},
\end{equation}
have been introduced. Now we apply the standard change of dependent variable
\begin{equation}
\label{eq:m_2:Tchihandrius}
\Phi=w\exp\left[-\frac{1}{2}\int_{z_0}^zp(s)\rmd s\right],
\end{equation}
which transforms~\eqref{eq:m_1:rational} into its reduced form
\begin{equation}
\label{eq:m_2:normal}
w''=r(z)w,\qquad r(z)=-q(z)+\frac{1}{2}p'(z)+\frac{1}{4}p(z)^2,
\end{equation}
where the coefficient $r(z)$ is given by
\begin{equation}
\label{eq:rr}
\begin{split}
r(z)&=-\frac{3}{16} \left(\frac{1}{(
   z+\gamma)^2}+\frac{1}{(z+1)^2}+\frac{1}{(z-1)^2}\right
   )+\frac{\alpha^2\beta}{(z+\alpha-\rmi\alpha\sqrt{\beta})^2(z+\alpha+\rmi\alpha\sqrt{\beta})^2}\\
   & +\frac{\alpha  (\alpha  (\beta +1) \gamma +4 \beta
   )+15 z^3+9 z^2 (2 \alpha +\gamma )+\alpha  z (3
   \alpha  (\beta +1)+10 \gamma )}{8
   \left(z^2-1\right) (z+\gamma ) \left(z+\alpha-\rmi \alpha 
   \sqrt{\beta } \right) \left(z+\alpha+\rmi \alpha 
   \sqrt{\beta }\right)}.
   \end{split}
\end{equation}
Equation~\eqref{eq:m_2:normal} has six regular singular points
\begin{equation}
z_{1,2}=\pm 1,\quad z_{3,4}=-\alpha\pm \rmi\alpha\sqrt{\beta},\quad z_5=-\gamma,\quad z_{6}=\infty,
\end{equation}
where the singularities $z_1,\dots,z_5$ are the poles of the second order, and
degree of infinity is $2$.  The respective differences of
exponents at singularities are following
\begin{equation}
\Delta_1=\Delta_2=\frac{1}{2},\quad \Delta_3=\Delta_4=0,\quad \Delta_5=\frac{1}{2},\quad \Delta_6=\frac{5}{2}.
\end{equation}
In  order to avoid their confluences  the conditions
\begin{equation}
\label{eq:m_2:param}
\gamma\notin\{\pm 1,\alpha\pm\rmi\alpha\sqrt{\beta}\},\quad \beta\neq-\frac{(\alpha\pm1)^2}{\alpha^2}
\end{equation}
must be satisfied. Since $\beta$ and $\alpha$ are real positive parameters only the
first condition remains valid. From this, we have the following exclusions for
the energy
\begin{equation}
E\notin \left\{-g(l_1(m_1+m_2)\pm l_2m_2), -gl_1(m_1\pm \rmi \sqrt{\frac{m_1}{m_2}}m_2)\right\}.
\end{equation}
Now, we can prove the following 
\begin{lemma}
The differential Galois group of equation~\eqref{eq:m_2:normal} satisfying
\eqref{eq:m_2:param} is $\operatorname{SL}(2,\C)$.
\end{lemma}
\begin{proof}
As $\Delta_{3,4}=0$, local solutions in a vicinity of  $z_*=z_3$, and
$z_*=z_4$, contain a logarithmic term. Two linearly independent solutions $w_1$
and $w_2$ of~\eqref{eq:m_2:normal} have the following forms
\begin{equation}
\label{eq:m_2:solutions}
w_1(z)=(z-z_*)^\rho f(z),\qquad w_z(z)=w_1(z)\ln(z-z_*)+(z-z_*)^\rho h(z),
\end{equation}
where $f(z)$ and $h(z)$ are holomorphic at $z_*$, and  $f(z_*)\neq 0$.  The
local monodromy matrix
corresponding to continuation of the matrix of fundamental solutions along a small loop
encircling $z_*$ counterclock-wise has the following form
\[
\vM_*=\begin{pmatrix}
-1& -2\pi\rmi \\
0& -1\\
\end{pmatrix},
\]
for details consult~\cite{Maciejewski:02::}. A subgroup of
$\operatorname{SL}(2, \C)$ generated by a non-diagonal triangular matrix cannot
be finite, and thus also differential Galois group cannot be finite. Moreover,
this matrix is non-diagonalisable. Thus also differential Galois group G of
this equation cannot be a subgroup of dihedral group that correspond to the
second case of the Kovacic algorithm.  Thus the only possibilities are that
$\scG$ is the full triangular group or $\operatorname{SL}(2,\ C)$.

In order to check the first possibility we apply the first case of the Kovacic
algorithm. If it is satisfied, then equation~\eqref{eq:m_2:normal} has an
exponential solution. First, for singular points $z_i$ we calculate  sets of
exponents $E_i$ of local solutions
\[
E_i:=\left\{(1\pm\Delta_i)/2\right\}, \qquad \text{for}\quad i=1,\dots, 6.
\]
 Thus, we have
 \begin{equation}
 \label{eq:m_2:auxiliary_sets}
 E_1=E_2=\left\{\frac{3}{4},\frac{1}{4}\right\},\quad 
 E_3=E_4=\left\{\frac{1}{2},\frac{1}{2}\right\},\quad 
 E_5=\left\{\frac{3}{4},\frac{1}{4}\right\},\quad 
 E_6=\left\{\frac{7}{4},-\frac{3}{4}\right\}.
 \end{equation}
 Next, according to the algorithm we look for elements 
 $e=(e_1,e_2,e_3,e_4,e_5,e_6)$ of Cartesian product $E=\prod_{i=1}^6E_i$ such that
 \begin{equation}
 d(e):=e_6-\sum_{i=1}^5e_i  \in \N_0.
 \end{equation}
 In our case there exists only one element of $E$ satisfying this condition, namely
 \[
 e=\left\{\frac{1}{4},\frac{1}{4},\frac{1}{2},\frac{1}{2},\frac{1}{4},\frac{7}{4}\right\},
 \]
 for which $d(e)=0$. Now we pass to the third step of the Kovacic algorithm. We
 look for a polynomial $P\neq 0$ of degree $d(e)$, such that it is a solution
 of the following differential equation
 \begin{equation}
 \label{eq:m_2:poly}
 P''+2\omega P'+ (\omega'+\omega^2-r)P=0,
 \end{equation}
 where 
 \[
 \omega=\sum_{i=1}^5\frac{e_i}{z-z_i},
 \]
 and $r$ is given in~\eqref{eq:rr}. In the considered case, we have $P=1$, so
 equation~\eqref{eq:m_2:poly} gives equality
 \[
 z+\alpha(1+\beta)=0,
 \]
 which cannot be satisfied for arbitrary $z$.
\renewcommand{\qedsymbol}{$\blacksquare$}
\end{proof}
\end{proof}

\section{Conclusions}

Although the obstructions to integrability obtained with the Morales--Ramis theory
are one of the strongest known, the frequent obstacle in its application is finding a
particular solution for a given dynamical system. The classical double pendulum,
both three- and two-dimensional, has no obvious solutions which could be used
to explicitly linearize the equations of motion and allow for determination of the
differential Galois group. Thus, despite the numerical evidence
\cite{Stachowiak:06::,Ivanov:99::} and theoretical
\cite{Burov:02::,Dullin:94::} work, a proof of
Liouvillian non-integrability still eludes us.

The present work is a step towards proving the conjecture that the
three-dimensional double pendulum has no additional first
integrals. Otherwise one would expect them to be present also in some
restrictions of the original system. Note also, that the finite set of energies
for which our non-integrability result does not hold does not leave any hope
for practical solving of the equations of motion, as that requires a global
first integral, independent of energy.

A peculiar feature of both models, which allows for such a concise analysis, is
the fact that the characteristic exponents of the variational equations do not
depend on the physical parameters. Since the differential Galois group depends
on the exponents, this is a huge simplification. Usually a mechanical system admits
first integrals for special values of parameters because the group depends on
them essentially. In this sense, the models considered here are quite exceptional.

\section*{Acknowledgement}
The work has been supported by grant No. DEC-2013/09/B/ST1/04130 
of National Science Centre of Poland.

\end{document}